\documentclass[journal,onecolumn]{IEEEtran}
\usepackage{cite,amsmath,amssymb,graphicx}
\usepackage{amsthm}
\usepackage{tikz}
\usetikzlibrary{trees}
\usepackage{quantikz}

\theoremstyle{plain}

\theoremstyle{definition}

\theoremstyle{plain}
\newtheorem{theorem}{Theorem}
\title{Quantum Coordination without Conditioning under Restricted Information}

\author{Faisal~Shah~Khan
\thanks{F. S. Khan is with the Kenan Institute for Private Enterprise, UNC Chapel Hill, Chapel Hill, NC, USA (email: Faisal\_Khan@kenan-flagler.unc.edu).}%
}

\begin{document}

\maketitle
\begin{abstract}
We study which joint probability distributions can be implemented by distributed systems when agents have restricted information. Classical local operations cannot implement all distributions that are achievable with full conditioning on history. We show that shared quantum states can overcome this limitation even when using only separable states. In particular, separable states that are diagonal in a fixed basis allow certain joint distributions to be implemented through local measurements without requiring agents to condition on an inaccessible latent variable. These distributions cannot be implemented by classical local rules under the same information constraints. However, quantum implementations cannot reproduce fully adaptive dependence on past outcomes when those outcomes are observationally indistinguishable. Quantum state preparation therefore provides a partial operational substitute for perfect recall under restricted information.
\end{abstract}

\section{Introduction}

Many distributed information-processing tasks require multiple agents or subsystems to produce coordinated outcomes while operating with only limited information about previous events. A central challenge in such settings is determining which joint distributions can be physically realized through local operations when agents cannot fully observe or condition on past events.

There, Kuhn’s theorem~\cite{Kuhn1953} establishes that, under perfect recall, any mixed strategy can be equivalently implemented as a behavioral strategy: agents can achieve any desired outcome distribution by conditioning their local probabilistic actions only on the information available at each decision point. The consequences of imperfect recall and the resulting breakdown of the equivalence between mixed and behavioral strategies have been studied extensively in game theory and decision theory~\cite{PiccioneRubinstein1997,Wichardt2008}. Related questions have also been investigated in quantum game settings with imperfect recall~\cite{Frackiewicz2011}. When recall is imperfect—that is, when the available information does not fully reveal the relevant history—this equivalence breaks down. Certain distributions that admit a latent-variable representation can no longer be implemented by any collection of local conditional rules, because the latent variable (or past outcomes) lies outside the agents’ information. The obstruction is therefore operational rather than purely mathematical: the distribution exists, but it cannot be realized through local actions under the given information constraints. This perspective is closely related to the role of correlation devices in game theory, where correlated behavior may be generated through an external randomization mechanism that is not directly available to individual decision makers~\cite{Aumann1974}.

The distributions of interest admit a latent-variable representation and can therefore be generated by classical shared randomness. However, under restricted information this randomness cannot be exploited by any collection of local conditional rules, because the latent variable lies outside every agent’s available information at each stage. The quantum construction works by encoding this same classical shared randomness into a shared separable state that is diagonal in a fixed basis, so that the target joint distribution can be recovered through local measurements that respect the information partition, without any party needing to condition on the inaccessible history.

In this work, we show that quantum systems can overcome this classical limitation for a broad class of distributions. We consider a model in which outcomes are generated by performing local measurements on a shared quantum state, with the choice of measurement at each stage depending only on the locally available information. Within this framework, we show that even separable quantum states can enable the physical implementation of joint distributions that are unattainable by classical local operations under the same information restrictions. The key mechanism is state preparation: separable states that are diagonal in a fixed basis encode latent coordination variables in a distributed manner, allowing target distributions to be recovered through local measurements without requiring any party to condition on inaccessible history. Thus, quantum state preparation can compensate for the breakdown of Kuhn's equivalence under restricted information without requiring entanglement or basis-dependent measurement asymmetries. More broadly, the result contributes to a growing body of work showing that useful quantum-information-processing tasks need not rely on entanglement~\cite{Bennett1999}.

At the same time, quantum implementations remain fundamentally constrained by the information structure. Because the shared state is fixed in advance, quantum systems cannot reproduce fully adaptive, history-dependent behavior when past outcomes are observationally indistinguishable under the available information. Thus, while quantum implementations extend the set of physically realizable distributions, they do not restore the full power of perfect recall.

Beyond game-theoretic settings, restricted-information coordination problems also arise in distributed quantum technologies, including fault-tolerant quantum computing, where syndrome processors may not possess the complete syndrome history required for globally coordinated decoding decisions. Potential applications to distributed quantum error correction are discussed in the conclusion.

The remainder of the paper is organized as follows. In Section II we formalize the classical setting and the limitation arising from restricted information. Section III introduces the quantum model. Section IV develops the general implementability result and its limitations. We conclude in Section V with a discussion of applications and open questions.

\section{Classical Correlation Models under Information Constraints}\label{ClasFram}

We formalize the classical setting using a probabilistic description of correlation generation under restricted information. Consider a sequential process in which outcomes \(X_1, \dots, X_n\) are generated at each stage \(k\) based only on the available information \(Y_k = f_k(X_{<k})\), where \(X_{<k} = (X_1, \dots, X_{k-1})\). When \(Y_k = X_{<k}\), agents have perfect recall of the full history; otherwise, information is restricted.

A local classical model consists of conditional distributions \(P_k(x_k \mid y_k)\) for each stage \(k\) and each possible information state \(y_k\). The resulting joint distribution factorizes as
\begin{equation}
P(x_1, \dots, x_n) = \prod_{k=1}^n P_k(x_k \mid y_k).
\label{eq:classical-factor}
\end{equation}
Thus, all correlations are generated exclusively through conditioning on the available information variables \(Y_k\).

Global randomization can be modeled by a latent variable \(S\) with distribution \(p(s)\), where each realization \(s\) determines a complete sequence \((x_1(s), \dots, x_n(s))\). The induced distribution is
\begin{equation}
P^\star(x_1, \dots, x_n) = \sum_{s:\, (x_1(s),\dots,x_n(s))=(x_1,\dots,x_n)} p(s).
\label{eq:global}
\end{equation}
Importantly, \(S\) is inaccessible at every stage and therefore cannot be used for local conditioning.

\subsection{Kuhn’s Theorem and Conditioning on History}

Kuhn’s theorem~\cite{Kuhn1953} (see also~\cite{OsborneRubinstein1994}) connects this probabilistic formulation to extensive-form games. Under perfect recall (\(Y_k = X_{<k}\)), any mixed strategy can be implemented by a behavioral strategy: at each information set \(I\) (corresponding to a value of \(Y_k\)), the behavioral strategy \(\sigma(a|I)\) is obtained by conditioning the distribution over complete histories on those that reach \(I\):
\begin{equation}
\sigma(a|I) = \frac{\sum_{h \in H(I,a)} p(h)}{\sum_{h \in H(I)} p(h)},
\end{equation}
where \(H(I)\) is the set of histories reaching \(I\) and \(H(I,a)\) is the subset in which action \(a\) is chosen.

Equivalently, in probabilistic terms, any joint distribution admits the factorization
\begin{equation}
P(x_1, \dots, x_n) = \prod_{k=1}^n P(x_k \mid x_{<k}).
\label{eq:perfect-recall}
\end{equation}
When information is restricted, however, the variables \(Y_k\) do not fully specify the past. In this case, there exist distributions \(P^\star\) generated by global randomization (Eq.~\eqref{eq:global}) that cannot be reproduced by any collection of local conditional rules of the form~\eqref{eq:classical-factor}. This failure of local implementability is the classical limitation we revisit in the quantum setting.

\subsection{Illustrative Example}\label{illex1}

Consider a simple binary case with outcomes \(X_1, X_2 \in \{0,1\}\) and restricted information in which \(Y_2\) is constant (the second stage has no access to \(X_1\)). The information structure can be visualized as the tree in Figure~\ref{fig:binary_tree}, where the dashed line indicates that the two second-stage nodes are indistinguishable.

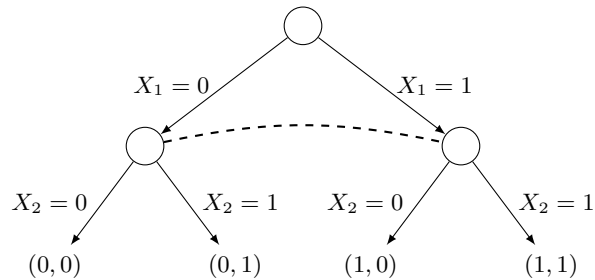
\begin{figure}[h]
\centering
\begin{tikzpicture}[
  level 1/.style={sibling distance=42mm},
  level 2/.style={sibling distance=24mm},
  level distance=16mm,
  decision/.style={circle, draw, inner sep=1.5pt, minimum size=5mm},
  edge from parent/.style={draw,-latex},
  every node/.style={font=\small}
]
\node[decision] (root) {}
  child {
    node[decision] (n1) {}
    child {
      node {$ (0,0) $}
      edge from parent node[left=2pt] {$X_2=0$}
    }
    child {
      node {$ (0,1) $}
      edge from parent node[right=2pt] {$X_2=1$}
    }
    edge from parent node[left=2pt] {$X_1=0$}
  }
  child {
    node[decision] (n2) {}
    child {
      node {$ (1,0) $}
      edge from parent node[left=2pt] {$X_2=0$}
    }
    child {
      node {$ (1,1) $}
      edge from parent node[right=2pt] {$X_2=1$}
    }
    edge from parent node[right=2pt] {$X_1=1$}
  };
\draw[dashed, thick] (n1) to[bend left=12] (n2);
\end{tikzpicture}
\caption{Tree representation of the binary illustrative example. The dashed line indicates that the two second-step nodes are indistinguishable under the available information, so the local rule generating \(X_2\) cannot depend on the realized value of \(X_1\).}
\label{fig:binary_tree}
\end{figure}

Let the target distribution be
\[
P^\star(0,1) = P^\star(1,0) = \frac12, \qquad P^\star(0,0) = P^\star(1,1) = 0.
\]
This distribution arises from global randomization with a uniform latent variable \(S \in \{0,1\}\):
\[
(X_1, X_2) = 
\begin{cases}
(0,1) & \text{if } S=0, \\
(1,0) & \text{if } S=1.
\end{cases}
\]

Under the given information constraint, any classical local model must factorize as \(P(x_1,x_2) = P_1(x_1) P_2(x_2)\). Setting \(P_1(0) = a\) and \(P_2(0) = b\) yields
\[
P(0,1) = a(1-b), \qquad P(1,0) = (1-a)b.
\]
Imposing \(P(0,1) = P(1,0) = \frac12\) forces \(a,b,1-a,1-b > 0\), which in turn implies \(P(0,0) = ab > 0\) and \(P(1,1) = (1-a)(1-b) > 0\), contradicting the target \(P^\star\). Therefore, \(P^\star\) cannot be implemented by any classical local rules under restricted information.

This example illustrates the core classical limitation: even when a distribution admits a latent-variable representation, it cannot be realized by local conditional rules if the latent variable is inaccessible through the agents’ information.

\section{Quantum Extension of the Classical Model}\label{QuanExt}

We now extend the classical framework of Section~\ref{ClasFram} to the quantum setting. The sequential structure and information constraints remain unchanged: at each stage \(k\), the outcome \(X_k\) is generated based only on the available information \(Y_k\).

The key difference is that correlations are no longer generated through local conditioning on past outcomes. Instead, they are encoded directly in a shared multipartite quantum state \(\rho\). At stage \(k\), the outcome \(X_k\) is obtained by performing a measurement on the \(k\)-th subsystem, where the choice of measurement (specified by a collection of operators \(\{M^{(k)}_{x_k}(y_k)\}_{x_k}\)) depends only on the local information \(y_k\). The resulting joint distribution is
\begin{equation}
P(x_1,\dots,x_n)
=
\operatorname{Tr}\!\Bigl[
\bigl(M^{(1)}_{x_1}(y_1) \otimes \cdots \otimes M^{(n)}_{x_n}(y_n)\bigr)\rho
\Bigr].
\label{eq:quantum-model}
\end{equation}
Correlations between outcomes are embedded in a shared state and revealed operationally through local measurements, without requiring agents to condition on past history.

Importantly, this model is designed to preserve the informational constraints of the classical setting: at each stage \(k\), the choice of measurement depends only on the available information \(Y_k\), and no signaling or adaptive quantum memory is introduced across stages. In this sense, Eq.~\eqref{eq:quantum-model} provides a natural quantum extension of the classical local model, differing only in the mechanism by which correlations are generated.

The central question is whether this quantum resource enlarges the set of implementable distributions under the same information constraints---in particular, whether distributions that cannot be written as \(\prod_k P_k(x_k \mid y_k)\) can nevertheless be realized by local quantum measurements.

\section{Quantum Implementability under Restricted Information}

We now formalize how separable quantum states can implement joint distributions that are unattainable by classical local rules under restricted information.

In the classical setting (Section~\ref{ClasFram}), correlations arise solely through conditioning on the available information variables \(Y_k\). Under perfect recall this allows any joint distribution to be implemented locally; under imperfect recall, many distributions generated by global randomization (including those with latent-variable representations) cannot be reproduced by rules of the form \(\prod_k P_k(x_k \mid y_k)\). The quantum model of Section~\ref{QuanExt} provides an alternative: correlations are encoded directly in a shared state \(\rho\) and revealed through local measurements (Eq.~\eqref{eq:quantum-model}), without requiring conditioning on past outcomes.

Related examples of coordination under restricted information have previously appeared in the context of imperfect-recall quantum games~\cite{Khan2025}. From a game-theoretic perspective, Theorem~\ref{thm1} may be viewed as a quantum realization of a correlation device whose latent variable remains inaccessible to the agents' information structure. The present result identifies a more general mechanism underlying such examples. Rather than relying on entanglement or other nonclassical correlation resources, implementability follows from the ability of a shared separable state to encode latent coordination variables that are inaccessible through the agents' information structure.

\subsection{Quantum Implementation of Hidden Coordination}
We first identify a broad class of distributions that admit a classical latent-variable representation yet cannot be implemented by classical local rules under restricted information. The obstruction is not the absence of a latent variable, but the inability to condition on it when it lies outside the information structure \(Y_k\).

The mechanism uses separable states that are diagonal in a fixed basis.

\begin{theorem}\label{thm1}
Consider a joint distribution of the form
\[
P^\star(x_1,\dots,x_n)
=
\sum_s p(s)\prod_{k=1}^n P_k(x_k \mid s, y_k),
\]
where \(s\) is a latent variable that determines the correlations but is inaccessible through the information variables \(Y_k\). Then \(P^\star\) can be implemented by local quantum measurements on a shared separable state, without making \(s\) accessible through \(Y_k\).
\end{theorem}
\begin{proof}
Let \(\{\ket{s}\}\) be an orthonormal basis indexed by \(s\). Define the separable state
\[
\rho = \sum_s p(s)\bigotimes_{k=1}^n \ket{s}\bra{s}^{(k)}.
\]
Each subsystem carries the classical label \(s\) in a distributed, locally inaccessible form. For each \(k\), define a POVM \(\{M^{(k)}_{x_k}(y_k)\}_{x_k}\) that is diagonal in the \(\{\ket{s}\}\) basis and satisfies
\[
\bra{s} M^{(k)}_{x_k}(y_k) \ket{s} = P_k(x_k \mid s, y_k).
\]
(For concreteness, one may take \(M^{(k)}_{x_k}(y_k) = \sum_s P_k(x_k \mid s, y_k) \ket{s}\bra{s}\).) These operators are positive and sum to the identity, hence form a valid POVM.

The resulting joint distribution is
\[
P(x_1,\dots,x_n)
=
\operatorname{Tr}\!\Bigl[
\bigl(M^{(1)}_{x_1}(y_1) \otimes \cdots \otimes M^{(n)}_{x_n}(y_n)\bigr)\rho
\Bigr]
= \sum_s p(s)\prod_{k=1}^n \bra{s} M^{(k)}_{x_k}(y_k) \ket{s}.
\]
By construction this equals \(P^\star(x_1,\dots,x_n)\). Thus \(P^\star\) is implemented by local quantum measurements on a shared separable state without granting access to \(s\).
\end{proof}

Theorem~\ref{thm1} shows that, under restricted information, state preparation can replace conditioning on an inaccessible latent variable. As a concrete illustration, apply Theorem~\ref{thm1} to the same target distribution from Section~\ref{illex1} using the classically correlated separable state
\[
\rho_{AB} = \frac12 |0\rangle\langle0|_A \otimes |0\rangle\langle0|_B + \frac12 |1\rangle\langle1|_A \otimes |1\rangle\langle1|_B
\]
together with the following local measurements (respecting the information constraint that \(Y_2\) is constant) implements \(P^\star\):
\begin{itemize}
    \item Stage 1 (Alice): measure subsystem \(A\) in the computational basis to obtain \(X_1\).
    \item Stage 2 (Bob): measure subsystem \(B\) in the computational basis and output the flipped result (\(X_2 = 1 -\) outcome).
\end{itemize}
This construction uses only classically correlated (diagonal) states and requires \textit{no quantum discord}. It exactly reproduces the target anti-correlated distribution even though Bob has no access to Alice’s outcome or to the latent variable \(S\).

Notably, this quantum advantage does not require quantum discord. As shown by Wei and Zhang~\cite{WeiZhang2017}, even purely classical correlations can yield an operational advantage when properly encoded in a shared quantum state and read out via the Born rule. This reinforces that the mechanism in Theorem~\ref{thm1} is quite general: classical shared randomness, when distributed through a quantum state, suffices to overcome the limitations imposed by restricted information.

\subsection{Limitations of Quantum Implementations}

Quantum implementations remain limited because the shared state \(\rho\) is fixed in advance. At each stage \(k\), the choice of measurement can depend only on the locally available information \(Y_k\); agents cannot adapt their measurements based on distinctions in history that are hidden behind \(Y_k\).

They can pre-encode useful correlations through state preparation, but they cannot replicate the full power of dynamic conditioning that perfect recall would allow.

Theorem~\ref{thm1} therefore shows that quantum state preparation acts as a partial substitute for perfect recall: it extends the set of implementable distributions beyond classical local rules, yet leaves a fundamental gap.

\section{Conclusion}

We have studied the implementability of joint distributions under restricted information. In the classical setting, Kuhn's theorem establishes that global randomization can be implemented through local conditional rules when agents have perfect recall. When recall is imperfect, however, this equivalence breaks down: certain distributions cannot be realized by any collection of local rules, even when they admit a latent-variable representation.

This work shows that quantum systems can overcome this classical limitation. Theorem~\ref{thm1} establishes that even separable quantum states can implement joint distributions that are unattainable by classical local rules under the same information constraints. The essential mechanism is state preparation: latent coordination variables can be encoded in a shared state and recovered through local measurements without requiring agents to condition on inaccessible history.

At the same time, quantum implementations remain limited by the fact that the shared state is fixed in advance. They therefore provide only a partial substitute for perfect recall: they enlarge the set of implementable distributions under restricted information, but they cannot restore fully adaptive conditioning on realized history.

These results highlight a clear distinction between classical and quantum mechanisms for coordination under restricted information. Characterizing the precise boundary between classical and quantum implementability under restricted information remains an important open problem. In particular, two questions arise from Theorem~\ref{thm1}:
\begin{itemize}
   \item All explicit constructions in this paper yield distributions that admit a latent-variable representation of the form in Theorem~\ref{thm1}. Is every distribution implementable by separable quantum states under restricted information representable in this way, or do separable quantum states permit a strictly larger class of implementable distributions?
    \item More generally, what is the full set of distributions that can be implemented quantum-mechanically under restricted information?
\end{itemize}

While the questions above are primarily theoretical, the framework also has potential relevance to practical quantum technologies. As an illustration, we now discuss an application to distributed syndrome decoding in fault-tolerant quantum computing.

\subsection{Application to Distributed Syndrome Decoding}

In large-scale fault-tolerant quantum computing architectures, syndrome extraction and decoding are frequently distributed across multiple hardware modules or spatially separated regions of a quantum processor. In such settings, no individual syndrome processor possesses the complete global syndrome history. Instead, each processor operates under an information constraint in which its locally available data \(Y_k\) constitutes only a restricted view of the syndrome. This situation is formally analogous to the restricted-information model studied in the preceding sections: the globally relevant error hypothesis is not fully observable at every local decision point.

To illustrate how Theorem~\ref{thm1} can be applied in this context, consider a minimal toy model consisting of three syndrome processors, labeled A, B, and C. Let \(S \in \{0,1\}\) denote a globally consistent binary hypothesis about the underlying error (for example, which of two possible global error patterns is compatible with the observed syndrome). A coordinated recovery strategy requires the processors to output correction decisions satisfying
\begin{equation}
(X_A, X_B, X_C) = (S, 1-S, S)
\end{equation}
with probability one. That is, processors A and C apply corrections consistent with hypothesis \(S\), while processor B applies the complementary action.

Under the given information constraints, the local view \(Y_k\) available to each processor does not uniquely determine the value of \(S\). Consequently, any classical local decoding rule must take the product form
\[
P(x_A, x_B, x_C) = \prod_{k \in \{A,B,C\}} P_k(x_k \mid y_k).
\]
As shown by the binary example in Section~II-B, no choice of such local conditional distributions can realize the target joint distribution over correction outcomes. The processors cannot condition their actions on the inaccessible global hypothesis \(S\), and any attempt to match the marginals on the coordinated triples necessarily produces positive probability on the undesired outcomes \((0,0,0)\) and \((1,1,1)\).

By contrast, the separable-state construction of Theorem~\ref{thm1} yields an immediate implementation that respects the same information constraints. Prepare the three-partite separable state
\begin{equation}
\rho = \frac12 |0\rangle\langle0|_A \otimes |0\rangle\langle0|_B \otimes |0\rangle\langle0|_C 
     + \frac12 |1\rangle\langle1|_A \otimes |1\rangle\langle1|_B \otimes |1\rangle\langle1|_C.
\end{equation}
Each processor \(k\) performs a local POVM \(\{M^{(k)}_{x_k}(y_k)\}_{x_k}\) whose elements are diagonal in the computational basis and satisfy
\[
\langle s | M^{(k)}_{x_k}(y_k) | s \rangle = P_k(x_k \mid s, y_k),
\]
where the conditional probabilities \(P_k(x_k \mid s, y_k)\) are chosen so that the output bit \(x_k\) equals \(S\) for processors A and C and equals \(1-S\) for processor B. Because the choice of measurement depends only on the locally available information \(y_k\), no processor is required to condition on the value of \(S\) or on any part of the syndrome history that lies outside its information set. The joint distribution over the correction decisions \((X_A, X_B, X_C)\) is then exactly the desired coordinated distribution.

This construction shows that a pre-shared separable quantum state can encode a global syndrome hypothesis in distributed form, enabling locally operating processors to coordinate on a consistent recovery action without any processor needing direct access to the full syndrome history. While the model above is deliberately minimal, it demonstrates a coordination primitive that may be relevant to modular or distributed decoding architectures in which communication latency, bandwidth, or hardware constraints prevent the exchange of complete syndrome information among decoding units.

Two limitations should be noted. First, because the state \(\rho\) is prepared in advance, the construction can coordinate only on hypotheses whose distribution is fixed at preparation time; it cannot implement fully adaptive decoding that would require processors to distinguish between syndrome histories that remain observationally indistinguishable under their local information \(Y_k\). Second, realizing a concrete performance advantage in a realistic setting would require embedding the mechanism into a specific quantum error-correcting code and comparing the resulting logical error rate against that achievable by classical local decoding under identical information constraints. These questions lie beyond the scope of the present work.

\bibliographystyle{IEEEtran}
\bibliography{references}

\end{document}